\documentclass{article}
\usepackage[utf8]{inputenc}
\usepackage[T1]{fontenc}
\usepackage{amsfonts} 
\usepackage{stmaryrd} 
\usepackage{amssymb} 
\usepackage{csquotes} 
\MakeOuterQuote{"}
\usepackage{amsmath} 
\usepackage{pdfpages} 
\usepackage{verbatim} 
\usepackage[labelformat=empty]{caption}
\usepackage{float} 
\usepackage{multirow} 
\usepackage{fancyhdr} 
\usepackage{framed} 
\usepackage{longtable} 
\usepackage{tikz} 
\usepackage{graphicx}
\usepackage{caption}
\usepackage{subfigure}
\usepackage{color}
\usepackage[stable]{footmisc} 
\usepackage{algorithm}
\usepackage{algorithmicx}
\usepackage{algpseudocode}
\newtheorem{theorem}{Theorem}
\newtheorem{definition}{Definition}
\newtheorem{proof}{Proof}
\newtheorem{corollary}{Corollary}
\title{RuleRunner technical report}
\author{Alan Perotti, Guido Boella: University of Turin}

\date{}
\begin{document}
   \maketitle


\section{Introduction} \label{sec:rw}

Runtime verification (RV) of a given correctness property $\phi$ (often formulated in linear temporal logic LTL \cite{ltl}) aims at determining the semantics of $\phi$ while executing the system under scrutiny; a monitor is defined as a device that reads a finite trace and yields a certain verdict \cite{rv}. Runtime verification may work on finite (terminated), finite but continuously expanding, or on prefixes of infinite traces. A monitor may control the current execution of a system (online) or analyse a recorded set of finite executions (offline). 
There are many semantics for finite traces: FLTL \cite{fltl}, RVLTL \cite{rvltl}, LTL3 \cite{ltl3}, LTL$\pm$ \cite{ltl+} just to name some. Since LTL semantics is based on infinite behaviours, the issue is to close the gap between properties specifying infinite behaviours and finite traces. There exist several RV systems, and they can be clustered in three main approaches, based respectively on rewriting, automata and rules.


\section{RuleRunner} \label{sec:rr}
RuleRunner is a rule-based online monitor observing finite but expanding traces and returning an FLTL verdict. As it scans the trace, RuleRunner mantains a state composed by rule names (for reactivating the rules), observations and formulae evaluations.

\begin{algorithm}
	\caption{Preprocessing and Monitoring Cycle}
	\label{alg:monitor}
\begin{algorithmic}[1]
	\State Parse the LTL formula in a tree
	\State Generate evaluation rules, reactivation rules and the initial state
	\While{new observations exist}
	\State Add observations to state
	\State Compute truth values using evaluation rules
	\State Compute next state using reactivation rules
	\If {state contains SUCCESS or FAILURE}
    \State \Return return SUCCESS or FAILURE respectively
	\EndIf
	\EndWhile
\end{algorithmic}
\end{algorithm}

In a nutshell, RuleRunner's behaviour (Algorithm \ref{alg:monitor}) is the following: in the preprocessing phase, RuleRunner encodes an LTL formula in a rule system. The rule system verifies the compliance of a trace w.r.t. the encoded property by entering a monitoring loop, composed by observing a new cell of the trace and computing the truth value of the property in the given cell. If the property is irrevocably satisfied or falsified in the current cell, RuleRunner outputs a binary verdict. If this is not the case, another monitoring iteration is entered, and -like in RuleR- undecided formulae trigger the reactivation of the corresponding monitoring rule. FLTL semantics guarantees that, if the trace ends, the verdict in the last cell of the trace is binary. \\
It is worth stressing how RuleRunner's approach is {\em bottom-up}, forwarding truth values from mere observations to the global property. RuleRunner does not keep a [multi]set of alternatives, as it is rooted in matching the encoding of the formula with the actual observations, computing the unique truth value of every subformula of the property, and carrying along a single state composed of certain information.\\

\begin{definition}
A RuleRunner system is a tuple $\langle R_E, R_R, S\rangle$, where $R_E$ ({\em Evaluation Rules}) and $R_R$ ({\em Reactivation Rules}) are rule sets, and $S$ (for {\em State}) is a set of active rules, observations and truth evaluations.\\
\end{definition}

We will define the rules in more detail in the following subsections; however, due to the lack of space, we will omit some technical details in order to keep the focus on the overall approach and the various components' interaction.\\

\subsection{Evaluation and reactivation rules}

RuleRunner accepts formulae $\phi$ generated by the following grammar:
$$\phi ::= true \mid \ a \mid \ !a \mid \phi \vee \phi \mid \phi \wedge \phi \mid \phi U \phi \mid X\phi \mid W\phi \mid \Diamond \phi \mid \Box \phi \mid END$$
$a$ is treated as an atom and corresponds to a single observation in the trace. We assume, without loss of generality, that temporal formulae are in negation normal form (NNF), i.e. negation operators pushed inwards to propositional literals and cancellations applied. $W$ is the weak negation operator. {\em END} is a special character that is added to the last cell of a trace to mark the end of the input stream.\\

An evaluation rule for $\phi$ is formed from an antecedent (body) and a consequent (head). The antecedent is a conjunction of literals, one of them being the rule name $R[\phi]$, and the others being truth evaluations $[\psi]V$, with $V \in \{T,F,?\}$ and $\psi$ being a subformula of $\phi$. The consequent is a single atom yielding a truth evaluation for $\phi$. Reactivation rules have one single atom as antecedent and a conjunction of atoms as consequent. The left-hand side of a reactivation rule is an undecided truth evaluation, the right-hand side a list of rule names. For example, consider the rules introduced in the previous section:
$$R[\Diamond a], [a]T \rightarrow [\Diamond a]{T}$$
$$R[\Diamond a], [a]F \rightarrow [\Diamond a]{?}$$
are evaluation rules (their only output is a truth evaluation), while
$$[\Diamond a]? \rightarrow R[a], R[\Diamond a]$$
is a reactivation rule, binding truth values in one cell to rule activation in the next cell.
The concept of rule activation is like introduced in RuleR. A rule is active if the rule name $R[\phi]$ is in the state $S$ of the RuleRunner system. For each active rule, if the condition part evaluates to true for the current cell, then the head of the rule is added to the current state. As introduced in Algorithm \ref{alg:monitor}, the verification loop alternately triggers evaluation and reactivation rules: the evaluation rules are used to compute the truth value of the property in the current cell, and the reactivation rules to define what rules are active in the following state. \\	


Each evaluation rule for $\phi$ corresponds to a single cell of the evaluation table for the main operator of $\phi$. Evaluation tables are three-valued truth tables (as introduced by Lukasiewitz and Kleene \cite{three}) with further annotations. \\
The tables in Figure 2 give the example for the $\Diamond$ and $\vee$ operators:\\

\begin{figure}[h!]
	\centering
	\includegraphics[scale=0.65]{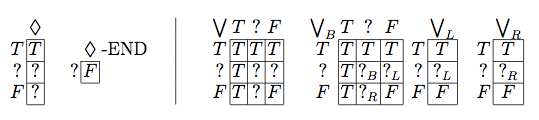}
	\caption{Fig.2: Evaluation tables for $\Diamond$, truth and evaluation tables for $\vee$}
\end{figure}

The three cells on the left define the run-time behavior of $\Diamond \phi$ given the truth value of $\phi$ in the current cell: $\Diamond \phi$ is true if $\phi$ is true, undecided otherwise. The single cell marked with $\Diamond - END$ represents one extra-rule, triggered only if the end of the trace has been reached: in that case, if the truth value of $\Diamond \phi$ is undecided, it is mapped to false. Intuitively, this is done since there is 'no future' left to satisfy $\phi$; moreover, it mirrors the concept of 'forbidden rules' in RuleR.\\


The right-hand side of Figure 2 reports the three-valued truth value table for $\vee$ and the evaluation tables for disjunction. $?_L$,$?_R$ and $?_B$ read, respectively, undecided {\em left}, {\em right}, {\em both}. For example, $?_L$ means that the future evaluation of the formula will depend on the left disjunct only, since the right one failed: the $V_R$ is, in fact, a unary operator. This allows the system to 'ignore' future evaluations of the right disjunct and it avoids the need to 'remember' the fact that the right disjunct failed: since that information is time-relevant (the evaluation failed at a given time, but it may succeed in another trace cell), keeping it in the system state and propagating it through time could cause inconsistencies.\\

The complete set of evaluation tables is reported in Figure 3:


\begin{figure}[h!]
	\includegraphics[scale=0.3]{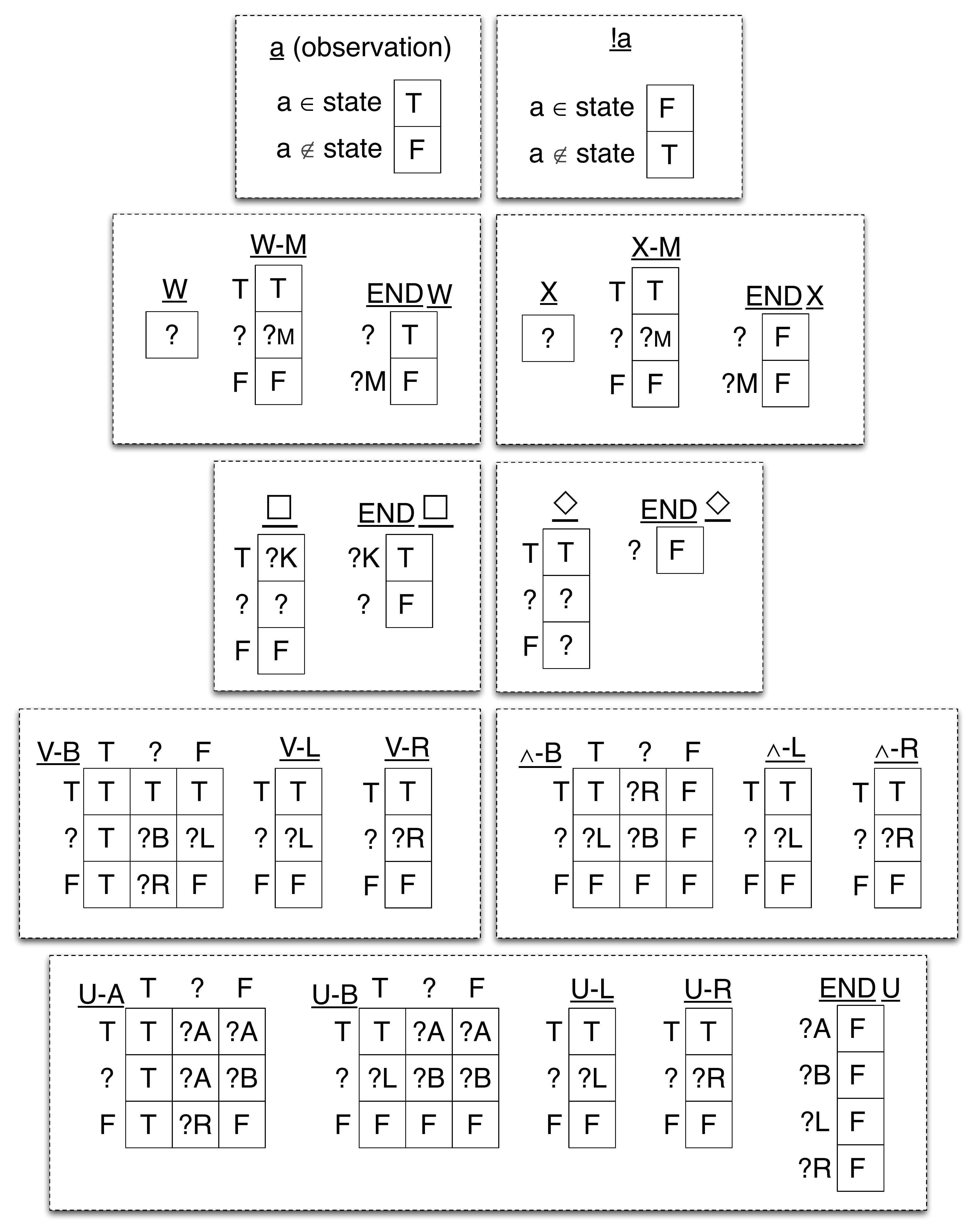}
	\caption{Fig.3: evaluation tables}
\end{figure}

\pagebreak

\begin{algorithm}[h!]
	\caption{Generation of rules}
	{\footnotesize
\begin{algorithmic}[1]
	\Function{Initialise}{$\phi$}
	\State $op \gets$ main operator
	\State
	\Comment{Apply recursively to subformula(e)}
	\If {$op \in \{\Box,\Diamond,X,W\}$}
	\State $\langle R_E^1, R_R^1, S^1 \rangle \gets$ Initialise($\psi^1$)
	\State $R_E \gets R_E^1$; $R_R \gets R_R^1$; 
	\ElsIf {$op \in \{\vee,\wedge,U\}$}
	\State $\langle R_E^1, R_R^1, S^1 \rangle \gets$ Initialise($\psi^1$)
	\State $\langle R_E^2, R_R^2, S^2 \rangle \gets$ Initialise($\psi^2$)
	\State $R_E \gets R_E^1 \cup R_E^2$; $R_R \gets R_R^1 \cup R_R^2$; 
	\Else 
	\State $R_E \gets \emptyset$; $R_R \gets \emptyset$;
	\EndIf
	\State
	\Comment{Compute and add evaluation rules for main operator}
	\State $Cells \gets$op's-evaluation-tables
	\ForAll{cell $\in$ Cells}
	\State Convert cell to single rule $r_e$, substituting formula names
	\State $R_E \gets R_E \cup r_e$
	\EndFor
	\If {$\phi$-is-main-formula}
	\State $R_E \gets R_E \cup ([\phi]T \rightarrow SUCCESS)$
	\State $R_E \gets R_E \cup ([\phi]F \rightarrow FAILURE)$
	\EndIf
	\State
	\Comment{Compute initial state for this subsystem}
	\If {$op = a$} $S \gets R[a]$	
	\ElsIf {$op = !a$} $S \gets R[!a]$
	\ElsIf {$op \in \{\vee,\wedge\}$} $S \gets S^1 \cup S^2 \cup R[\phi]B$
	\ElsIf {$op = U$} $S \gets S^1 \cup S^2 \cup R[\phi]A$
	\ElsIf {$op \in \{\Box,\Diamond\}$} $S \gets S^1 \cup R[\phi]$
	\ElsIf {$op \in \{X,W\}$} $S \gets R[\phi]$
	\EndIf
	\State
	\Comment{Compute and add reactivation rules for main operator}
	\If {$op \in \{\vee,\wedge\}$} $R_R \gets R_R \cup ([\phi]?Z \rightarrow R[\phi]?Z)$, for $Z \in {L,R,B}$
	\ElsIf {$op = U$} $R_R \gets R_R \cup ([\phi]?Z \rightarrow R[\phi]?Z,S^1,S^2)$, for $Z \in {A,B,L,R}$
	\ElsIf {$op \in \{\Box,\Diamond\}$}  $R_R \gets R_R \cup ([\phi]? \rightarrow R[\phi],S^1)$
	\ElsIf {$op \in \{X,W\}$} $R_R \gets R_R \cup ([\phi]? \rightarrow R[\phi]M,S^1) \cup ([\phi]?M \rightarrow R[\phi]M)$
	\EndIf
	\State
	\Comment{Return computed system}
	\State \Return $\langle R_E, R_R, S \rangle$
 	\EndFunction
\end{algorithmic}
}
\end{algorithm}

The generation of evaluation and reactivation rules is summarised in Algorithm 2.
The algorithm visits the parsing tree in post-order. The system is built incrementally, starting from the system returned by the recursive call(s). As introduced in the previous subsection, a RuleRunner system is defined as $\langle R_E, R_R, S \rangle$, the symbols meaning {\em evaluation rules, reactivation rules} and {\em state} respectively. If $\phi$ is an observation (or its negation), an initial system is created, including two evaluation rules, no reactivation rules and the single $R[\phi]$ as initial state. If $\phi$ is a conjunction or disjunction, the two systems of the subformulae are merged, and the conjunction/disjunction evaluation rules, reactivation rule and initial activation are added. The computations are the same if the main operator is $U$, but the reactivation rule will have to reactivate the monitoring of the two subformulae. Formulae with $X$ or $W$ as main operator go through two phases: first the formula is evaluated to undecided, as the truth value can't be computed until the next cell is accessed. Special evaluation rules force the truth value to false (for $X$) or true (for $W$) if no next cell exists. Then, at the next iteration, the reactivation rule triggers the subformula: this means that if $X\phi$ is monitored in cell $i$, $\phi$ is monitored in cell $i+1$. $\phi$ is then monitored independently, and the $X\phi$ (or $W\phi$) rule enters a 'monitoring state' (suffix M in the table), simply mirroring $\phi$ truth value and self-reactivating. Finally, $\Box$ and $\Diamond$ constantly reactivate themselves and their subformula, unless they are (respectively) falsified and verified at runtime, or forced to true or false when the trace ends.

\subsection{Example}
Consider the formula $a \vee \Diamond b$ and the trace $[c - a - b,d - b]$. RuleRunner creates the following rule system:\\
\begin{minipage}[h]{.54\linewidth}
\vspace{0pt}
	{\footnotesize
	EVALUATION RULES
	\begin{itemize}
		\item $R[a]$, $a$ is observed $\rightarrow$ $[a]{T}$
		\item $R[a],$ $a$ is not observed $\rightarrow$ $[a]{F}$
		\item $R[b]$, $b$ is observed $\rightarrow$ $[b]{T}$
		\item $R[b],$ $b$ is not observed $\rightarrow$ $[b]{F}$
		\item $R[\Diamond b]$, $[b]{T}$  $\rightarrow$ $[\Diamond b]{T}$
		\item $R[\Diamond b]$, $[b]{?}$  $\rightarrow$ $[\Diamond b]{?}$
		\item $R[\Diamond b]$, $[b]{F}$  $\rightarrow$ $[\Diamond b]{?}$
		\item $[\Diamond b]?$, $[END]$ $\rightarrow$ $[\Diamond b]{F}$
		\item $R[a\vee\Diamond b]B$, $[a]{T}$, $[\Diamond b]{T}$ $\rightarrow$ $[a\vee\Diamond b]{T}$
		\item $R[a\vee\Diamond b]B$, $[a]{T}$, $[\Diamond b]{?}$ $\rightarrow$ $[a\vee\Diamond b]{T}$
		\item $R[a\vee\Diamond b]B$, $[a]{T}$, $[\Diamond b]{F}$ $\rightarrow$ $[a\vee\Diamond b]{T}$
		\item $R[a\vee\Diamond b]B$, $[a]{?}$, $[\Diamond b]{T}$ $\rightarrow$ $[a\vee\Diamond b]{T}$
		\item $R[a\vee\Diamond b]B$, $[a]{?}$, $[\Diamond b]{?}$ $\rightarrow$ $[a\vee\Diamond b]{?B}$
		\item $R[a\vee\Diamond b]B$, $[a]{?}$, $[\Diamond b]{F}$ $\rightarrow$ $[a\vee\Diamond b]{?L}$
		\item $R[a\vee\Diamond b]B$, $[a]{F}$, $[\Diamond b]{T}$ $\rightarrow$ $[a\vee\Diamond b]{T}$
		\item $R[a\vee\Diamond b]B$, $[a]{F}$, $[\Diamond b]{?}$ $\rightarrow$ $[a\vee\Diamond b]{?R}$
		\item $R[a\vee\Diamond b]B$, $[a]{F}$, $[\Diamond b]{F}$ $\rightarrow$ $[a\vee\Diamond b]{F}$
		\item $R[a\vee\Diamond b]L$, $[a]{T}$ $\rightarrow$ $[a\vee\Diamond b]{T}$
		\item $R[a\vee\Diamond b]L$, $[a]{?}$ $\rightarrow$ $[a\vee\Diamond b]{?L}$
		\item $R[a\vee\Diamond b]L$, $[a]{F}$ $\rightarrow$ $[a\vee\Diamond b]{F}$
		\item $R[a\vee\Diamond b]R$, $[\Diamond b]{T}$ $\rightarrow$ $[a\vee\Diamond b]{T}$
		\item $R[a\vee\Diamond b]R$, $[\Diamond b]{?}$ $\rightarrow$ $[a\vee\Diamond b]{?R}$
		\item $R[a\vee\Diamond b]R$, $[\Diamond b]{F}$ $\rightarrow$ $[a\vee\Diamond b]{F}$
		\item $[a\vee\Diamond b]T$ $\rightarrow$ $SUCCESS$
		\item $[a\vee\Diamond b]F$ $\rightarrow$ $FAILURE$
	\end{itemize}
}
\end{minipage}
\begin{minipage}[h]{.41\linewidth}
	\vspace{0pt}
	{\footnotesize
		REACTIVATION RULES	
		\begin{itemize}
			\item $[\Diamond b]? \rightarrow R[b], R[\Diamond b]$
			\item $[a \vee \Diamond b]?B \rightarrow R[a \vee \Diamond b]B$
			\item $[a \vee \Diamond b]?L \rightarrow R[a \vee \Diamond b]L$
			\item $[a \vee \Diamond b]?R \rightarrow R[a \vee \Diamond b]R$
		\end{itemize}
\hrule	
\vspace*{3mm}	
INITIAL STATE	
$$R[a], R[b], R[\Diamond b], R[a \vee \Diamond b]B$$
\hrule 
\vspace*{3mm}
EVOLUTION OVER $[c - a - b,d - b]$ \\

\begin{tabular}{ |r||l|}
	\hline
state & $R[a], R[b], R[\Diamond b]B, R[a \vee \Diamond b]$ \\ \hline
+ obs & $R[a], R[b], R[\Diamond b]B, R[a \vee \Diamond b], c$ \\ \hline
eval & $[a]F, [b]F, [\Diamond b]?, [a \vee \Diamond b]?R$ \\ \hline
react & $R[b], R[\Diamond b], R[a \vee \Diamond b]R$ \\ \hline
\multicolumn{2}{c}{} \vspace*{-3mm}\\
	\hline
state & $R[b], R[\Diamond b], R[a \vee \Diamond b]R$ \\ \hline
+ obs & $R[b], R[\Diamond b], R[a \vee \Diamond b]R,a$ \\ \hline
eval & $[b]F, [\Diamond b]?, [a \vee \Diamond b]?R$ \\ \hline
react & $R[b], R[\Diamond b], R[a \vee \Diamond b]R$ \\ \hline
\multicolumn{2}{c}{} \vspace*{-3mm}\\
	\hline
state & $R[b], R[\Diamond b], R[a \vee \Diamond b]R$ \\ \hline
+ obs & $R[b], R[\Diamond b], R[a \vee \Diamond b]R,b,d$ \\ \hline
eval & $[b]T, [\Diamond b]T, [a \vee \Diamond b]T, SUCCESS$ \\ \hline
STOP & PROPERTY SATISFIED \\ \hline
\end{tabular}
}
\end{minipage}

\pagebreak
 The behaviour of the runtime monitor is the following: 
\begin{itemize}
	\item In the first cell, $c$ is observed. $a$ is false, $b$ is false, $\Diamond b$ is undecided. The global formula is undecided, but since the trace continues the monitoring goes on.
	\item In the second cell, $a$ has to be ignored (because the property required it to be observed in the previous cell); since $b$ is false again, $\Diamond b$ and $a \vee \Diamond b$ are still undecided
	\item In the third cell, $d$ is ignored but observing $b$ satisfies, in cascade, $b$, $\Diamond b$ and $a \vee \Diamond b$. The monitoring stops, signalling a success. The rest of the trace is ignored.
\end{itemize}
RuleRunner provides rich information about the 'verification status' of a property: in any iteration the state describes which subformulae are under monitoring and what the truth value is; when the monitoring ends, the state itself explains why the property was verified/falsified.

\subsection{Semantics}
RuleRunner implements the FLTL \cite{fltl} semantics; however, there are two main differences in the approach. Firstly, FLTL is based on rewriting judgements, and it has no constraints over the accessed cells, while RuleRunner is forced to complete the evaluation on a cell before accessing the next one. Secondly, FLTL proceeds top-down, decomposing the property and then verifying the observations; RuleRunner propagates truth values bottom up, from observations to the property. In order to show the correspondence between the two formalisms, we introduce the map function:
\begin{center} $map : $ Property $\rightarrow$ FLTL judgement \end{center}
The $map$ function translates the state of a RuleRunner system into a FLTL judgement, analysing the state of the RuleRunner system monitoring $\phi$. Since $\Box$ and $\Diamond$ are derivate operators and they don't belong to FLTL specifications, we omit them from the discussion in this section.\\

\begin{algorithmic}
	\Function{map}{$\phi$, State,index}
\If {$SUCCESS \in $ State} \Return $\top$
\ElsIf {$FAILURE \in $ State} \Return $\perp$
\ElsIf {$[\phi]T \in $ State} \Return $\top$
\ElsIf {$[\phi]F \in $ State} \Return $\perp$
\ElsIf {$[\phi]?S \in $ State} $aux \gets S$
\Else \ find {$R[\phi]S \in $ State}; $aux \gets S$
\EndIf
\If {$\phi = a$}
	\State \Return $[u,index \models a]_{F}$
\ElsIf {$\phi =  \ !a$}
	\State \Return $[u,index \models \neg a]_{F}$
\ElsIf {$\phi = \psi^1 .. \psi^2 \ and \  aux = L$}
\State \Return $map(\psi^1)$
\ElsIf {$\phi = \psi^1 .. \psi^2 \ and \ aux = R$}
\State \Return $map(\psi^2)$	
\ElsIf {$\phi = \psi^1 \vee \psi^2 \ and \ aux = B$}
	\State \Return $map(\psi^1) \sqcup map(\psi^2)$
\ElsIf {$\phi = \psi^1 \wedge \psi^2 \ and \ aux = B$}
	\State \Return $map(\psi^1) \sqcap map(\psi^2)$
\ElsIf {$\phi = \psi^1 U \psi^2 \ and \ aux = A$}
	\State \Return $map(\psi^2) \sqcup (map(\psi^1) \sqcap (map(X(\psi^1 U \psi^2))))$
\ElsIf {$\phi = \psi^1 U \psi^2 \ and \ aux = B$}
	\State \Return $map(\psi^2) \sqcap (map(X(\psi^1 U \psi^2)))$	next
\ElsIf {$\phi = X\psi \ and \ aux \not=M$}
	\State \Return $[u,index \models X\psi]_{F}$
\ElsIf {$\phi = W\psi \ and \ aux \not=M$}
	\State \Return $[u,index \models \bar{X}\psi]_{F}$
\ElsIf {$(\phi = X\psi \ or \ \phi = W\psi) \ and \ aux = M$}
	\State \Return $map(\psi)$	
\EndIf
 \EndFunction
\end{algorithmic}
The following table reports a simple example of an evolution of a RuleRunner step and the corresponding value computed by $map$. Let the property be $a \vee X b$ and the trace be $u = [b-b]$. The index is incremented when the reactivation rules are fired. 
\begin{table}[h!]
\begin{center}
\begin{tabular}{ |l|l|}
\hline 
{\bf State} & {\boldmath$map(a \vee X b)$} \\ \hline
{$R[a],R[Xb],R[a \vee X b]B$} & {$[u,0 \models a]_F \sqcup [u,0 \models Xb]_F$} \\ \hline
{$R[a],R[Xb],R[a \vee X b]B,b$} & {$[u,0 \models a]_F \sqcup [u,0 \models Xb]_F$} \\ \hline
{$R[a],R[Xb],R[a \vee X b]B,b,[a]F$} & {$\perp \sqcup [u,0 \models Xb]_F$} \\ \hline
{$R[a],R[Xb],R[a \vee X b]B,b,[a]F, [b]?M$} & {$\perp \sqcup [u,0 \models Xb]_F$} \\ \hline
{$R[a],R[Xb],R[a \vee X b]B,b,[a]F, [b]?M, [a \vee Xb]?R$} & {$[u,0 \models Xb]_F$} \\ \hline
{$R[b],R[Xb]M,R[a \vee X b]R$} & {$[u,1 \models b]_F $} \\ \hline
{$R[b],R[Xb]M,R[a \vee X b]R, b$} & {$[u,1 \models b]_F $} \\ \hline
{$R[b],R[Xb]M,R[a \vee X b]R, b, [b]T$} & {$\top $} \\ \hline
{$R[b],R[Xb]M,R[a \vee X b]R, b, [b]T, [Xb]T$} & {$\top $} \\ \hline
{$R[b],R[Xb]M,R[a \vee X b]R, b, [b]T, [Xb]T, [a \vee Xb]T$} & {$\top $} \\ \hline
{$SUCCESS$} & {$\top $} \\ \hline
\end{tabular}
\vspace*{2mm}
\caption{Fig.4: The $map$ function}
\end{center}
\end{table}
\vspace*{-5mm}
\begin{theorem} For any well-formed LTL formula $\phi$ over a set of observations, and for every finite trace $u$, for every intermediate state $s_i$ in RuleRunner's evolution over $u$ there exist a valid rewriting $r_j$ of $[u,0 \models \phi]_F$ such that $map(\phi) = r_j$. In other words, RuleRunner's state can always be mapped onto an FLTL judgement over $\phi$.
\end{theorem}
\begin{proof} The proof proceeds by induction on $\phi$:
	\begin{itemize}
		\item {\boldmath$\phi = a$}\\ If the formula is a simple observation, then the initial state is $R[a]$, and $map(R[a]) = [u,0 \models a]_{F}$. Adding observation to the state does not change the resulting FLTL judgement. If $a$ is observed, RuleRunner will add $[a]T$ to the state, and this will be mapped to $\top$. If $a$ is not observed, RuleRunner will add $[a]F$ to the state, and this will be mapped to $\perp$. So for this simple case, the evolution of RuleRunner's state corresponds either to the rewriting $[u,0 \models a]_{F} = \top$ (if $a$ is observed) or to the rewriting $[u,0 \models a]_{F} = \perp$ (if $a$ is not observed).
		\item {\boldmath$\phi = !a$}\\
		This case is analogous tho the previous one, with opposite verdicts.
		\item {\boldmath $\phi = \psi^1 \vee \psi^2$} \\		
		By inductive hypothesis, a RuleRunner system monitoring $\psi^1$ always corresponds to a rewriting of $[u,i\models \psi^1]$. The same holds for $\psi^2$. Let $\langle R_R^i, R_E^i, S^i \rangle$ be RuleRunner system monitoring the subformula $\psi^1$, with $i \in  \{ 1,2 \}$. A RuleRunner system encoding $\phi$ includes $R^1$ and $R^2$ rules and specific rules for $\psi^1 \vee \psi^2$ given the truth values of $\psi^1$ and $\psi^2$. The initial state is therefore $R[\psi^1 \vee \psi^2] \cup S^1 \cup S^2$, and this is mapped to $map(S^1) \sqcup map(S^2)$. By inductive hypothesis, this is a valid FLTL judgement. In each iteration, as long as the truth value of $\psi^1 \vee \psi^2$ is not computed, the state is mapped on $map(S^1) \sqcup map(S^2)$. When the propagation of truth values reaches $\psi^1 \vee \psi^2$, the assigned truth value mirrors the evaluation table for the disjunction. If either $\psi^1$ or $\psi^2$ is true, then $\phi$ is true, and $map(\phi) = \top$. This corresponds to the valid rewriting $map(S^1) \sqcup map(S^2) = \top$, given that we are considering the case in which there is a true $\psi^i$: $[\psi^i]T$ belongs to the state and $map(\psi^1) = \top$. The false-false case is analogous. In the $?_B$ case, the mapping is preserved, and this is justified by the fact that both $\psi^1$ and $\psi^2$ are undecided in the current cell, therefore $map(\psi^i) \not = \top,\perp$, therefore $map(\psi^1) \sqcup map(\psi^2)$ could not be simplified. In the $?_L$ case, we have that $[\psi^2]F$, therefore $map(\psi^2) = \perp$. The FLTL rewriting is $map(\psi^1) \sqcup map(\psi^2) = map(\psi^1)$, and this is a valid rewriting since $map(\psi^1) \sqcup map(\psi^2) = map(\psi^1) \sqcup \perp = map(\psi^1)$. The $?_R$ case is symmetrical.
	\item {\boldmath $\phi = \psi^1 \wedge \psi^2$} \\	
	Same as above, with the evaluation table for conjunction on the RuleRunner side and the $\sqcap$ operator on the FLTL judgement side.
	\item {\boldmath $\phi = X\psi$} \\
	A RuleRunner system encoding $X\phi$ has initial state $R[X\phi]$, which is mapped on $[u,0 \models X\psi]_F$. Then, if the current cell is the last one, $R[X\phi]$ evaluates to $[X\phi]F$, and the corresponding FLTL judgement is $\perp$. If another cell exists, $R[X\phi]$ evaluates to $[X\phi]?$ (with the same mapping). When the reactivation rules are triggered, $[X\phi]?$ is substituted by $R[X\psi]M,R[\psi]$. Over this state, $map(X\psi) = map (\psi)$, and the index is incremented since reactivation rules were fired. Therefore, the FLTL rewriting is $[u,i \models X\psi] = [u,i+1 \models \psi]$, and this is a valid rewriting.
	\item {\boldmath $\phi = W\psi$} \\
	This case is like the previous, but if the current cell is the last then $R[W\psi]$ evolves to $[W\psi]T$; the mapping is rewritten from $[u,i \models W\psi]$ to $\top$, and this is a valid rewriting if there is no next cell.
	\item {\boldmath $\phi = \psi^1 U \psi^2$} \\
	The initial RuleRunner system includes rules for $\psi^1$, $\psi^2$ and for the $U$ operator. As long as $R[\psi^1 U \psi^2]A$ is not evalued, $map(\psi^1 U \psi^2) = map(\psi^2) \sqcup (map(\psi^1) \sqcap (map(X(\psi^1 U \psi^2))))$, that is, the standard one-step unfolding of the 'until' operator as defined in FLTL. When a truth value for the global property is computed, there are several possibilities. The first one is that $\psi^2$ is true and $\psi^1 U \psi^2$ is immediately satisfied. RuleRunner adds $[\psi^1 U \psi^2]T$ to the state and $map(\phi) = \top$; this corresponds to the rewriting $map(\psi^2) \sqcup (map(\psi^1) \sqcap (map(X(\psi^1 U \psi^2)))) = \top \sqcup (map(\psi^1) \sqcap (map(X(\psi^1 U \psi^2)))) = \top$, which is a valid rewriting. The case for $[\psi^1]F$ and $[\psi^2]F$ is analogous. The $?_A$ case means that the evaluation for the until is undecided in the current trace, and is mapped on the standard one-step unfolding of the until operator in FLTL. The $?_B$ case implicitly encode the information that 'the until cannot be trivially satisfied anymore', and henceforth the FLTL mapping is $map(\psi^1) \sqcap (map(X(\psi^1 U \psi^2)))$. The cases for $?_L$ and $?_R$ have the exact meaning they had in the disjunction and conjunction cases. For instance, if $[\psi^1]F$ and $[\psi^2]?$, RuleRunner adds $[\psi^1 U \psi^2]?R$ to the state, and for the obtained state $map(\phi) = map(\psi^2)$. The sequence of FLTL rewriting is $map(\psi^2) \sqcup (map(\psi^1) \sqcap (map(X(\psi^1 U \psi^2)))) = map(\psi^2) \sqcup (\perp \sqcap (map(X(\psi^1 U \psi^2)))) = map(\psi^2) \sqcup \perp = map(\psi^2)$.
	\end{itemize}
\end{proof}
\begin{corollary}
	RuleRunner yields a FLTL verdict.
\end{corollary}
\begin{proof}
	RuleRunner is always in a state that can be mapped on a valid FLTL judgement; therefore, when a binary truth evaluation for the encoded formula is given, this is mapped on the correct binary evaluation in FLTL. But since for such trivial case the $map$ function corresponds to an identity, the RuleRunner evaluation is a valid FLTL judgement. The fact that RuleRunner yields a binary verdict is guaranteed provided that the analysed trace is finite, thanks to end-of-trace rules.
\end{proof}


\subsection{Complexity}
RuleRunner generates several rules for each operator, but this number is constant, as it corresponds to the size of evaluation tables plus special rules (like the SUCCESS one). The number of rules corresponding to $\phi \vee \psi$, for instance, does not depend in any way on the nature of $\phi$ or $\psi$, as only the final truth evaluation of the two subformulae is taken into account. The preprocessing phase creates the parse tree of the property to encode and adds a constant number of rules for each node (subformula). Then, during the runtime verification, for each cell of the trace the system goes through all rules exactly once. This is guaranteed by the fact that the rules are added in a precise order, assuring pre-emption for rules evaluating simpler formulae. This is simply implemented by the post-order visit of the parsing tree, as shown in Algorithm (2). 
This strict ordering among formulae guarantees that, when the set of rules regarding (i.e.) $\phi U \psi$ is considered, both $\phi$ and $\psi$ have been evaluated already, and their evaluations belong to the current state of the system.\\
Therefore, the complexity of the system is inherently polynomial. This complexity is not in contrast with known exponential lower bounds for the temporal logic validity problem, as RuleRunner deals with the satisfiability of a property on a trace, thus tacking a different problem from the validity one. This kind of distinction is also mentioned in \cite{trover}. In general, the exponential nature of many approaches (e.g., Büchi Automata) arises from listing all possible combinations of observations before matching them with the actual ones. We avoid this by computing a single, distributed state, containing only the certain (and therefore single) truth value of every subformula, computed after the observation phase.\\


\section{Prototyping} \label{sec:code}
A Java implementation is available at www.di.unito.it/$\sim$perotti/RV13.jnlp.\\ The prototype requires the user to enter a well-formed LTL formula, which is parsed before unlocking the actual verification settings. Traces can be either manually typed in the GUI, randomly generated (after setting some parameters) or loaded from files. The output is binary for the second and third case, and verbose for the first one.


\bibliographystyle{splncs}
\bibliography{ruleRunnerArXiv}
%
%


\end{document}